\documentclass[runningheads]{llncs}
\usepackage[T1]{fontenc}
\usepackage[hidelinks]{hyperref}
\usepackage{tikz}
\usepackage{amsmath}
\usepackage{enumitem} 
\usepackage{graphicx}

 \definecolor{lime}{HTML}{A6CE39}
 \DeclareRobustCommand{\orcidicon}{
 	\begin{tikzpicture}
 	\draw[lime, fill=lime] (0,0)
 	circle [radius=0.16]
 	node[white] {{\fontfamily{qag}\selectfont \tiny ID}};
 	\draw[white, fill=white] (-0.0625,0.095)
 	circle [radius=0.007];
 	\end{tikzpicture}
 }
 \renewcommand{\orcidID}[1]{\href{https://orcid.org/#1}{\ensuremath{\orcidicon}}}

\begin{document}
\title{Sequential Elimination Voting Games}
\author{
    Ulysse Pavloff
    \and
    Tristan Cazenave
    \and
    Jérôme Lang
}
\authorrunning{U. Pavloff et al.}
\institute{LAMSADE, Université Paris-Dauphine, PSL, CNRS, CEA List\\
\email{ulysse.pavloff@cea.fr}
}
\maketitle              
\begin{abstract}
Voting by sequential elimination is a low-communication voting protocol: voters play in sequence and eliminate one or more of the remaining candidates, until only one remains. 
While the fairness and efficiency of such protocols have been explored, the impact of strategic behaviour has not been addressed. We model voting by sequential elimination as a game. Given a fixed elimination sequence, we show that the outcome is the same in all subgame-perfect Nash equilibria of the corresponding game, and is polynomial-time computable. We measure the loss of social welfare due to strategic behaviour, with respect to the outcome under sincere behaviour, and with respect to the outcome maximizing social welfare.
We give tight bounds for worst-case ratios, and show using experiments that 
the average impact of manipulation can be much lower than in the worst case. 

\end{abstract}

\section{Introduction}

The classical view of voting is centralized: voters submit their votes to a central authority, which computes the outcome and sends it back to the voters. However, such a centralized voting protocol is not suitable in (typically low-stake) contexts where there are few voters and possibly many candidates. On the one hand, asking each voter to report a ranking over all candidates puts too much burden on them, incurring a cost that may exceed the benefit they can draw from taking part to the vote; on the other hand, asking them to rank or approve a small set of candidates (possibly even one) can be highly inefficient if the number of candidates exceeds the number of voters, as votes can be scattered in such a way that no candidate gets more than a handful of votes, making the collective decision quite arbitrary. 

For this reason, a recent line of research has focused on the design of {\em low-communication decentralized protocols}. In such protocols, there is no central authority and voters are asked to report a very small amount of information. A specific low-communication voting protocol is defined in \cite{bouveret_voting_2017}: we start from a sequence of agents, whose length is the number of candidates minus one; this sequence is either fixed from the beginning or obtained by a randomization process. At each step, the agent designated by the sequence eliminates one of the candidates that have not been eliminated yet. After the last step there remains one candidate, which is declared the winner. \cite{bouveret_voting_2017} explore the normative properties of voting by sequential elimination, and suggest to choose elimination sequences that maximize the expected Borda score of the winner, which guarantees some (weak) form of efficiency and fairness; they show how such sequences can be computed. However, they leave a strategic study for further research.

As a consequence of the Gibbard-Satterthwaite theorem, sequential elimination, seen as a voting rule, is manipulable except when the elimination sequence involves only one voter. 
Now, a sequential elimination protocol can also be seen as a turn-based game, where voters play by eliminating a candidate whenever it is their turn to play. Assuming perfect knowledge and perfectly rational agents, the suitable game-theoretic concept here is the {\em subgame perfect Nash equilibrium} (SPNE); since agents are assumed to have strict preferences over candidates, all SPNE have the same outcome (elected candidate).
Now, some questions come up:

\begin{enumerate}
    \item Is it easy to compute the SPNE outcome of a sequential elimination game? Does it have a nice characterization?
    \item What is the social welfare of the SPNE outcome of a sequential elimination game, compared to the social welfare of the optimal outcome? and compared to the social welfare obtained if all agents play sincerely?
\end{enumerate}

As for Question 1, we give two positive answers: the SPNE outcome is computable in polynomial time, and moreover it has a very simple characterization: it is the outcome obtained by reversing the sequence and assuming voters play sincerely. Such a characterization is reminiscent of equilibria in sequential resource allocation \cite{kohler_class_1971,kalinowski_strategic_2013}, although the setting is quite different. This has some interesting simple consequences: first, if the sequence is {\em palindromic} (invariant by symmetry), then the strategic outcome is the same as the sincere outcome, so that strategic behaviour has no impact. This might seem contradicting the fact that the voting process is manipulable; however, a manipulation is a deviation by a single agent; here just saying that if all agents play best responses, then their strategic behaviour will {\em globally} have no impact. 

As for Question 2, we have to choose a way of measuring social welfare. If agents have arbitrary utilities consistent with their ranking over candidates, as usual, the cost is unbounded: such impossibility results abound in the distortion-theoretic study of voting. However, if we take the Borda score of a candidate as a proxy for its social welfare (which is a classical assumption, {\em e.g.},
\cite{BranzeiCMP13,bouveret_voting_2017}),
we show that with $n$ voters and $m$ candidates, the corresponding {\em price of anarchy} (defined, as usual, as the worst-case ratio between the social welfare of the best outcome and that of the SPNE) is $\frac{O_{\rm{max}}-1+(n-1)(m-1)}{m-1}$, where $O_{\rm{max}}$ is the number of  occurrences of the voter who appears most frequently in the sequence. We also define the {\em sincerity ratio} as the worst-case ratio between the social welfare of the outcomes obtained if all voters are, respectively, sincere and strategic.
For this measure we found an upper bound in the order of $n$, which is somewhat reasonable as the number of agents is assumed to be low.

We complete our worst-case study by an average-case study via simulations, so as to measure in practice the price of strategic behaviour obtained when the voters' profiles follow a given distribution, for the choice of a few specific elimination sequences, and for several distributions over profiles.

The outline of the paper is as follows. We survey related work in Section \ref{sec:related}. 
In Section \ref{sec:spne} we give computation and characterization results about SPNE. 
In Section \ref{sec:social-welfare} we study the price of strategic behaviour: price of anarchy (Subsection \ref{subsec:poa}), sincerity ratio (Subsection \ref{subsec:bounds}). In \autoref{sec:average} we study the average case of these measures. Section \ref{sec:conclu} concludes.

\section{Related work}\label{sec:related}

Voting with low communication is a relatively new topic. Most of it is surveyed in 
\cite{BoutilierR16}. Somewhat related to our work is the notion of {\em distortion} in voting: if we assume that agents have cardinal utilities but can only report ordinal preferences upon which a voting rule is applied, what is the worst case ratio between the social welfare of the optimal candidate and that of the outcome of the voting rule used? The key question consists in deciding which trade-off between distortion and amount of communication is needed. A survey can be found in \cite{AnshelevichF0V21}. 

Our starting point is \cite{bouveret_voting_2017}, who defines a specific rule and communication protocol for voting, parameterized by an elimination sequence. They study its axiomatic properties and give an algorithm for identifying the sequence giving the result closest to the Borda rule.

Other specific low-communication rules have been recently defined in \cite{GrossAX17} and  \cite{ChenL020}.

Voting in stages where at each stage a voter takes an action has been the topic of a few works. 
{\em Stackelberg voting games} \cite{XiaC10a} assume that voters have complete knowledge of the others' preferences, and vote in a predefined sequence, observing the votes already cast before theirs. The main difference with our work is the form of the action taken by voters: in \cite{XiaC10a} they choose a complete vote (a ranking) and in the end a fixed voting rule is used to determine the outcome; in our setting, the actions available to the voters are candidate-eliminating actions. In \cite{DesmedtE10}, voters also act in sequence and choose either between their (sincere) vote or abstention, given that voting incurs a small penalty. 

\cite{Anbarci06} considers an elimination game 
like ours but with only two agents (and 7 alternatives, but the results extend to more) and the alternating sequence 121212, and 
characterizes the SPNE outcome, which our Theorem 1 generalizes (see Section \ref{sec:spne}).

Viewing strategic voting in a static context as a game is classical and dates back to \cite{farquharson_theory_1969}. The algorithmic aspects of strategic voting games with few voters have been recently explored \cite{ElkindGRS15}. A recent research trend focuses on voting dynamics, where voters are able to change their vote after gaining partial information about other votes: see \cite{Meir17} for a survey. \cite{BranzeiCMP13} and \cite{KavnerXia21}, determine the price of anarchy in iterative voting for a few specific voting rules. 



Elimination sequences are reminiscent of {\em picking sequences} used in fair division of indivisible goods. They have been given a game-theoretic analysis in  \cite{kohler_class_1971} and \cite{kalinowski_social_nodate}, which give a characterization and a computational analysis of SPNE. Their manipulation has been studied  in \cite{bouveret_general_2011,BouveretL14,TominagaTY16,Walsh16,AzizBLM17,AzizGT17}.

\section{Equilibria in sequential elimination games}\label{sec:spne}

A {\em sequential elimination (voting) game} consists of:

\begin{itemize}
    \item a set of players (or voters) $N =\{1, \dots , n\}$;
    \item a set of candidates $C = \{c_1 , \dots , c_m \}$; actions consist in eliminating a candidate that has not been eliminated yet.
    \item a preference profile $V=  (V_1,\dots,V_n)$; each $V_i$ is a linear order over $C$ expressing the preferences of voter $i$;
    \item an \textit{elimination sequence} for $n$ voters and $m$ candidates, that is,  a sequence $\pi = (\pi(1), \dots , \pi(m-1))$ in $N^{m-1} $, where $\pi(i)$ is the $i^{th}$ voter in the sequence.
\end{itemize}  

Sometimes we note $\succ_{V_i}$ instead of $V_i$. We use the abridged notation $abcd$ to denote the ranking $a \succ_{V_i} b \succ_{V_i} c \succ_{V_i} d$. Similarly, we write an elimination sequence (1,2,3,4) as well as 1234 with the same signification.
For $c \in C, r(c,V_i)$ denotes the rank of $c$ in $V_i$ (from 1 for the
best candidate to $m$ for the worst one). The Borda score
of $c$ for $V$ is $S_B(c,V)=\sum_{i=1}^{n}{m - r(c, V_i)}$.

\label{palindromic_def} A \textit{palindromic sequence} is an elimination sequence $\pi$ of size $m$ such that $\pi(i)=\pi(m-i)$. For example, $(1,2,3,3,2,1)$ is palindromic.

An elimination action is said to be {\em sincere} if the acting voter eliminates her least favorite candidate. A voter is sincere if all of her elimination actions are sincere. If all voters play sincerely in an  elimination game then we say that the elimination game is sincere. We assume complete information, which is rather standard in algorithmic game theory.

\begin{example}
Let $\pi=(1,2,3,1)$ and $V = (abcde, edcba, debca)$. 
If all voters are sincere then 1 eliminates $e$, then 2 eliminates $a$, then 3 eliminates $c$ and finally 1 eliminates $d$, so the winner will be $b$. 
\end{example} 

A {\em subgame-perfect Nash equilibrium} (SPNE) is defined in games that consist of multiple stages or subgames. 
A SPNE is a strategy profile $s^{*}$ with the property that in no subgame can any player $i$ do better by choosing a strategy
different from $s^{*}_i$, given that every other player $j$ adheres to $s^{*}_j$.
The {\em outcome} of a SPNE is the corresponding winning candidate.
Because voters have strict preferences over outcomes, all SPNE of the game are associated with the same outcome.

Our main characterization result is the following:

\begin{theorem} \label{firstTheorem}
A SPNE of a vote by sequential elimination has the same outcome as sincere voting with the sequence reversed.
\end{theorem}

\begin{proof}
By induction on $m$.
The case $m = 1$ is clear. Suppose the result
holds for $m-1$. We show that it holds for $m$.
Let $\pi$ be the elimination sequence.
Let us focus on any SPNE of the game. 
The last voter of the sequence $\pi$ is $\pi(m-1)$.
For every state of the game tree before stage $m-1$, the voters have no incentive to eliminate the least preferred candidate of $\pi(m-1)$. 
Indeed, if the least preferred candidate of $\pi(m-1)$ is one of the two remaining candidates when $\pi(m-1)$ votes, $\pi(m-1)$ will eliminate her. 
We then look at the rest of the sequence with the $m-1$ remaining candidates and we know by induction hypothesis that any SPNE gives the same outcome as sincere voting with the reversed sequence.
\end{proof}

This result generalizes the characterization of SPNE for the Strike game \cite{Anbarci06}. Also, it is reminiscent of a result in fair division of indivisible goods: when two agents play a round robin game (strict alternation), 
 the SPNE can be computed by simply reversing the policy and the preference ordering \cite{kohler_class_1971}; it was generalized in \cite{kalinowski_strategic_2013} to any picking sequence, but still for two agents; our Theorem \ref{firstTheorem}, which does not concern allocation but voting, holds for any number of agents.  

Let us look at an example before explaining the reasoning behind  Theorem \ref{firstTheorem} and the corollaries that stem from it.

\begin{example}\label{exampleFirstTheorem}
Let $n = 3$, $m = 4$, $V = (abcd, cbad, cadb)$ and $\pi=(1,2,3)$. 
If the voters are sincere then 1 eliminates $d$, 2 eliminates $a$ and 3 eliminates $b$ which gives us $c$ as the winner of the sincere vote.  Now, if the voters are strategic then 1 eliminates $c$, 2 eliminates $d$ and 3 eliminates $b$ which gives us $a$ as the winner. Indeed if 1 eliminates $c$, then 2 has the choice of eliminating $a$, $b$ or $d$. 
Knowing that 3 will eliminate $b$ if 2 does not, 2 only has to choose the candidate who will win between $a$ and $d$. So 2 chooses to eliminate $d$, then 3 eliminates $b$ and $a$ wins.

\begin{figure}
\centering
\begin{tikzpicture}[scale=0.81,
        level 1/.style = {sibling distance = 2cm,level distance = 1.2cm},
        level 2/.style = {sibling distance = 2.8cm,level distance = 1.3cm},
        level 3/.style = {sibling distance = 1.4cm,level distance = 1.5cm},
        rn/.style={circle, draw=gray!60, fill=gray!5, minimum size=5mm, inner sep=0pt,thick}, 
        blank/.style={draw=white,fill=white, minimum size=4mm,inner sep=0pt},
        squarednode/.style={rectangle, draw=gray!60, fill=white, inner sep=0pt,thick, minimum size=5mm},
]

\node[rn]{1}
    child{node[rn]{2} edge from parent node [blank]{a}}
    child{node[rn]{2} edge from parent node [blank]{b}}
    child{
        node[rn]{2}
        child{ [black]
            node[rn]{3} 
            child{ 
                node[squarednode]{d} edge from parent[->] node [blank]{b}
            }
            child{ 
                node[squarednode]{b} edge from parent[->] node [blank]{d}
            }
            edge from parent node [blank]{a}
        }
        child{ [black]
            node[rn]{3} 
            child{ 
                node[squarednode]{d} edge from parent[->,black] node [blank]{a}
            }
            child{ 
                node[squarednode,draw=orange!60]{a} edge from parent[->,orange] node [blank]{d}
            }
            edge from parent[orange] node [blank,above]{b}
        }
        child{
            node[rn]{3} 
            child{ 
                node[squarednode]{b} edge from parent[->,black] node [blank]{a}
            }
            child{
                node[squarednode,draw=red!60]{a} edge from parent[->,red] node [blank]{b}
            }
            edge from parent[red] node [blank]{d}
        }
        edge from parent[red] node [blank]{c}
    }
    child{ node[rn]{2} edge from parent node [blank]{d}   }
    
    ;

\end{tikzpicture}
\caption{Part of the game tree, Example \ref{exampleFirstTheorem}.}
\label{tree:exampleFirstTheorem}
\end{figure}
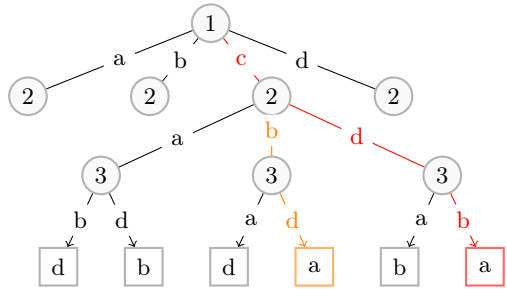

\autoref{tree:exampleFirstTheorem} depicts the part of the game tree where voter 1 eliminates $c$. Circled nodes represent  voters, boxed nodes represent the winner and each edge represents an action. The path leading to the outcome is in red. As we can see, voter 2 can indifferently choose between eliminating $b$ (orange path) or $d$, with the same outcome.
\end{example}

The idea behind \autoref{firstTheorem} is the following: suppose that in each final non-terminal state (each state having only terminal states as successors), each agent has a \textit{weakly dominant action}, i.e. an action that provides at least the same utility as all the others. 
We can then replace each final non-terminal state by the state that will be reached as a result of the agent's dominant strategy. By repeating this process indefinitely we can go back to the beginning, in this case we say that the game is solvable by \textit{within-state dominant-strategy backward induction (WSDSBI)}. 
We know from  \cite{xia_strategic_2011} that if a game is solvable by WSDSBI, then the outcome is unique.

\begin{corollary} \label{cor:palindromic}
If a sequence $\pi$ is palindromic, the outcome of the SPNE is the outcome of the sincere vote. 
\end{corollary}

\begin{corollary} 
\label{cor:nTimes}
If a voter occurs $q$ times in the elimination sequence, the outcome of the SPNE will not be one of her $q$ least preferred candidates.
\end{corollary}

Let us comment further on Corollary \ref{cor:nTimes}. 
In a sincere vote, if a voter $i$ votes only once in the sequence, when her turn comes she eliminates the candidate she least prefers among the remaining ones. Thus, if $i$ does not eliminate her least preferred candidate, it is because this candidate has already been eliminated. 
We can see that this also works if $i$ votes $q$ times. 
The only reason why $i$ does not eliminate some of her $q$ least preferred candidates is that 
they already have been eliminated by another voter.
Strategic voting being equivalent to sincere voting with the sequence reversed, this reasoning remains valid for strategic voting.

In Appendix \ref{appendix:sec:sincstratvoters}. we extend this characterization to games with both strategic and sincere voters.

\section{Social welfare worst-case analysis}\label{sec:social-welfare}
In this section we discuss the possible detrimental effects of strategic votes on social welfare. For this we introduce measures considering worst-case scenarios for social welfare. We first study the price of anarchy and then introduce a new measure more relevant to the specific type of game we consider.

\subsection{Price of Anarchy}\label{subsec:poa}

The \textit{price of anarchy} \cite{koutsoupias_worst-case_2009} is 
defined when social welfare is quantified. In the absence of explicit cardinal utilities, it is classical to use the Borda score as a measure of social welfare ({\em e.g.}, \cite{BranzeiCMP13,bouveret_voting_2017}), and from now on this is what we will do. A first interesting question is to evaluate the price of anarchy associated with a given elimination sequence.
The price of anarchy is the worst-case ratio (over all preference profiles) of the maximum social welfare over the one reached in the SPNE.
Given a preference profile $V$, let $a$  be the candidate with maximum Borda score and $b$ the winner in the SPNE.  Then
\begin{equation}\label{basicRatio}
    PoA(\pi) = \max_V R_{ab}(\pi,V), \;\mbox{where}\; R_{ab}(\pi,V)=\frac{S_B(a,V)}{S_B(b,V)}.
\end{equation}

For sequential elimination games, it is possible to evaluate this quantity by means of the following theorem: 

\begin{theorem}\label{theorem:price-anarchy}
With $n$ voters and $m$ candidates the price of anarchy of a sequence $\pi$ calculated with
Borda score is 
\begin{equation}\label{price-anarchy}
PoA(\pi) = \frac{O_{\rm{max}}-1+(n-1)(m-1)}{m-1}, 
\end{equation}
where $O_{\rm{max}}$ is the number of 
occurrences of the voter  who appears most frequently in the sequence.
\end{theorem}

\begin{proof}
We want to determine an upper bound for $R_{ab}$ over all profiles $V$ for a given sequence. To do so, we want to find a candidate $a$ with the highest Borda score possible while simultaneously having $b$ elected during the strategic vote. A necessary feature is that at least one voter should prefer $b$ to $a$ to ensure that $b$ can be elected during a strategic vote. Indeed, if every voter prefers $a$ to $b$, $b$ cannot be elected.
 
We call $x$ the voter preferring $b$ to $a$ ($b\succ_{V_x}a$). Now let us specify the necessary features of $V_x$ in order to maximize the ratio $R_{ab}$. To do so let us decompose $R_{ab}$ as follow :
\begin{equation}\label{eq:RabDecomposed}
    R_{ab}=\frac{S_B(a,V)}{S_B(b,V)}=\frac{S_B(a,V-\{V_x\}) + S_B(a,V_x)}{S_B(b,V-\{V_x\})+S_B(b,V_x) }
\end{equation}
To minimize the impact of a voter $x$ preferring $b$ to $a$ on the ratio $R_{ab}$, candidates $a$ and $b$ should be next to each other\footnote{Indeed, a larger ranking difference between $a$ and $b$ would be detrimental to the maximization of \eqref{eq:RabDecomposed}. Please remember that the higher $r(i,V_j)$ the least preferred is the candidate $i$ for the voter $j$.} in $V_x$: $r(a,V_x) - r(b,V_x) = 1$, which is equivalent to $S_B(b,V_x)=S_B(a,V_x)+1$.
To determine the rank of $a$ and $b$ in the preferences of $x$ we can write $R_{ab}$ as a function of $k=S_B(a,V_x)$ : $f(k) = \frac{\alpha+k}{\beta+k}$,  with $\alpha=S_B(a,V-\{V_x\})$ and $\beta=1+S_B(b,V-\{V_x\})$.

By construction, $a$ is the candidate with the highest Borda score. Thus we have $S_B(a,V) \geq S_B(b,V)$. Hence subtracting from both sides the quantity $S_B(b,V_x) (=S_B(a,V_x)+1)$ one gets $\alpha \geq \beta$.
This immediately implies that $f$ is a decreasing function of $k$. 
Therefore, no matter which voter ($x$) prefers $b$ over $a$, to maximize the ratio we must increase\footnote{Increasing the rank is equivalent to decreasing the Borda score: $S_B(i,V_j)=m-r(i,V_j)$.} the rank of $a$ and $b$ in her preferences, which results in $r(b,V_x) = m-O^{(\pi)}_x$ (higher limit of the rank $r(b,V_x)$  of an elected candidate according to corollary \ref{cor:nTimes}) 
and $r(b,V_x) = r(a,V_x)-1$. Hence, $S_B(a,V_x) = S_B(b,V_x)-1 = O^{(\pi)}_x-1$. 
It is clear that $S_B(a,V_x)$ is maximum when $x$ is the voter with the most occurrences in the sequence ($O^{(\pi)}_x=O_{\rm{max}}$).

Bearing this in mind we now determine an upper bound for $\alpha$ and lower bound for $\beta$.

$S_B(a,V-{V_x})$ is maximum if every voter except $x$ have $a$ as their favorite candidate which gives the following upper bound : $S_B(a,V-{V_x}) \leq (n-1)(m-1)$.
As for $\beta$ we know of the minimal value of an elected candidate thanks to corollary \ref{cor:nTimes}: it is equivalent to the sum of each candidate's occurrences in the sequence. Thus we have $\beta \geq m-1-O_{\rm max}$. This leads to the following upper bound for $R_{ab}$:
\begin{equation*}
\forall V, R_{ab}(\pi,V) \leq \frac{(n-1)(m-1)+O_{\rm{max}}-1}{m-1} \; .
\end{equation*}
Thus:
\begin{equation} \label{finalRatioPOA}
PoA(\pi) \leq \frac{(n-1)(m-1)+O_{\rm{max}}-1}{m-1} \; .
\end{equation}

\

It now remains to make sure that there exists a preference profile for which this bound is indeed reached. 
To this end we need to exhibit a set of preferences leading to $a$ being the optimal candidate and $b$ being elected in a strategic vote despite having the minimum Borda score for an elected candidate:
\begin{enumerate}[label=(\roman*)]
    \item For $b$ to have the minimum Borda score, the rank of $b$ shall be $r(b,V_i)=m-O^{(\pi)}_i, \forall i \in N$. 
    \item Next we know that the voter ($x$) with the highest number of occurrences in the sequence must have $a$ right after $b$ in her preferences (with $b\succ_{V_x}a$). So $r(a,V_x)=r(b,V_x)+1$.
    \item $a$ must be the favorite candidate of every voter except $x$. $\forall i \in N-\{x\}$, $r(a,V_i)=1$.
    \item The last requirement is that a candidate should not appear with a ranking worse than $b$ in the preferences of more than one voter otherwise $b$ will be eliminated.
\end{enumerate}
It is easy to exhibit a preference pattern meeting these four criteria. We do not do so here but we will for the more involved case of the sincerity ratio in Appendix \ref{appendix:sec:struct}. 
Hence the upper bound \eqref{finalRatioPOA} is reached, which completes the proof.
\end{proof}

\textit{Remark 1:}  The order of votes in the sequence is irrelevant to the price of anarchy. The only important element is the maximum number of occurrences of a voter.

\textit{Remark 2:}
We have identified four criteria for a preference profile to reach the price of anarchy. Among them (i) and (ii) impose that $a$ is one of the $O_x^{(\pi)}$ least preferred candidates  of voter $x$. According to Corollary \ref{cor:nTimes}, this implies that $a$ can never be elected for such a preference profile.
\

We have seen how strategic behavior can deteriorate social welfare compared to an ideal candidate. However, as Remark 2 indicates, this ideal candidate  cannot be elected may the vote be sincere or strategic. This leads us to define the \textit{sincerity ratio}, as we do next.

\subsection{Sincerity ratio}\label{subsec:bounds}

We now consider the sincerity ratio which, unlike the price of anarchy, uses as a reference point  the candidate elected if voters play sincerely in the sequential elimination game, and not the optimal candidate in the absolute sense. By measuring
the loss of social welfare when voters are strategic rather than sincere, this ratio is an appropriate tool to quantify the vulnerability of voting by sequential elimination to strategic behaviour.

Unlike the price of anarchy, there is no guarantee that the sincerity ratio should always be at least 1, and we will see that it is indeed the case: there are cases where selfish behavior is beneficial to social welfare.

Consider a $n$-voter preference profile $V$ over $m$ candidates, and a sequence $\pi$ of length $m-1$. 
For a given preference profile $V$, let $c$ and $b$ be the candidate elected with sincere voting and the candidate elected in the SPNE, respectively. 
We define the sincerity ratio as:
\begin{equation}\label{sincerityRatio}
    SR( \pi) = \max_V R_{cb}(\pi,V),  \;\mbox{where}\; R_{cb}(\pi,V) =\frac{S_B(c,V)}{S_B(b,V)}.
\end{equation}

\begin{example}
Let us come back to Example \ref{exampleFirstTheorem}. The candidate $c$ is the winner of the sincere vote, and $a$ the one of the strategic vote. $S_B(c,V)=7$ and $S_B(a,V)=6$, thus in this case $R_{cb}(\pi,V) = 7/6$. If we change the sequence $\pi$ to $\pi^{-1}=(3,2,1)$ then the result of the sincere vote for $\pi$ is equivalent to the result of the strategic vote for $\pi^{-1}$ and conversely. So $R_{cb}(\pi^{-1},V) = 6/7<1$ which means that in this case, strategic voting is beneficial to social welfare.
\end{example}

As illustrated by this example, $R_{cb}(\pi,V)$ is a measure of the relative impact on social welfare of sincere vote with respect to strategic vote.
$SR(\pi)$ defined in \eqref{sincerityRatio} is then just a measure of the maximal positive relative impact of sincere vote for a given sequence. One can estimate an upper bound for this quantity by means of the following theorem:

\begin{theorem}\label{theorem:sincerity-ratio}
With $n$ voters and $m$ candidates the sincerity ratio 
has the following upper bound:
\begin{equation}\label{sincerity-ratio}
SR(\pi) \leq \frac{O_{\rm{max}}+(n-1)(m-1)}{m}, 
\end{equation}
where $O_{\rm{max}}$ is the number of 
occurrences of the voter with the highest number of occurrences in the sequence.
\end{theorem}

This proof is similar to the one of \autoref{theorem:price-anarchy}. We highlight the main differences in Appendix \ref{appendix:sec:SRproof}. We then give a characterization of the structure of instances for which the worst case ratio is reached in Appendix \ref{appendix:sec:struct}. 

\section{Experimental validation}\label{sec:average}
Until now we
focused on worst case scenarios. 
We now consider in more details the
general characteristics of the distributions of the above studied measures.

\subsection{Experiment Setup}

We have seen that in the worst case, the sincerity ratio and the price of anarchy can be in the order of $n$. We would like to know to what extent the ratios $R_{ab}(\pi,V)$ from Equation \eqref{basicRatio} and $R_{cb}(\pi,V)$ from Equation \eqref{sincerityRatio} obtained on average are far from the worst cases. Also, since our lower and upper bounds do not coincide for the sincerity ratio (while they do for the price of anarchy), we would like to compute the exact sincerity ratio in some typical cases to see how far it is from the upper bound. We address these two questions by numerical experiments.

When defining the expected sincerity ratio as its average over all possible profiles, we first consider the {\em impartial culture assumption}, which is classical in social choice theory. As it has been often criticized as not being realistic, in a second step we consider the Mallows culture \cite{mallows_ranking_1957} for which voters' preferences are correlated. 

In a first series of experiments, we test exhaustively the  $(m!)^{n-1}$ different combinations of possible preference profiles for $(n,m)=(2,8)$ and $(3,7)$.
This gives the exact values of the worst-case and average sincerity ratios. In the second series, we choose larger values: $n=5$ and $m=10$. In this case we cannot be exhaustive and proceed by Monte-Carlo sampling.

In each table presenting an exhaustive test over all the profiles (\autoref{tab:expPOA1}, Table \ref{tab:expSR1}), for each chosen sequence we display:
\begin{itemize}
    \item the mean ratio 
    ($R_{ab}$ or $R_{cb}$), obtained by averaging over all profiles. We also show its standard deviation.
    \item the  worst case value of the ratio. 
\end{itemize}
In Tables \ref{tab:expSR1} and \ref{tab:expSR2} we also display  the worst-case theoretical upper bound of $SR(\pi)$ obtained in Equation \eqref{sincerity-ratio}. We denote this upper bound as U.B.
For each table using Monte-Carlo sampling (Tables \ref{tab:expPOA2} and \ref{tab:expSR2}), we show for each chosen sequence the mean value of the studied ratio under Mallows culture $M_{\phi}$ with different values of $\phi$, including Impartial Culture, obtained for $\phi = 1$.

For $(n,m)=(5,10)$ we tested Mallows with different values of the dispersion parameter $\phi$. As it turns out, from $\phi=0.1$ to $\phi=0.5$, the mean value of both $R_{ab}$ and $R_{cb}$ is always equal to 1 and the standard-deviations to 0. So it appears that in case of a highly correlated culture, the elected candidate is the candidate with the highest Borda Score. This is the case for both sincere and strategic voting.

\subsection{Price of Anarchy, in Practice}\label{subsec:poa_exp}
 
On \autoref{tab:expPOA1} we present for a few sequences the mean value of the ratio $R_{ab}$ \eqref{basicRatio} with its standard-deviation, and its maximum value reached ,which correspond to the $PoA$ \eqref{price-anarchy}. We computed these values after testing each one of the $8!$ 
and $(7!)^2$
different preference profiles.

\begin{table}
\centering
\caption{Mean value of $R_{ab}(\pi)$ as in \eqref{basicRatio}, standard-deviation, and maximum value reached}
\begin{tabular}{c|c|c|c}
$(n,m)$ & sequence & $\langle R_{ab} \rangle$ & $PoA$ \\ \hline
(2,8) & 1112221  & $1.02 \pm 0.05$ & $10/7$ \\
(2,8) & 1222111  & $1.03 \pm 0.07$ & $10/7$ \\
(2,8) & 1122111  & $1.06 \pm 0.1$ & $11/7$ \\
(3,7) & 123123  & $1.06 \pm 0.13$ & $13/6$ \\  
(3,7) & 123321 & $1.06 \pm 0.13$ & $13/6$ \\  
(3,7) & 111223  & $1.06 \pm 0.13$ & $14/6$ \\
(3,7) & 112233  & $1.07 \pm 0.13$ & $13/6$ \\

\end{tabular}
\label{tab:expPOA1}
\end{table}

The average value of $R_{ab}$ is close to 1 and its standard deviation is small.  
As the number of voters and the $PoA$ gets bigger, the mean value of $R_{ab}$ increases slightly. 
In order to try and confirm this trend we studied larger sequences presented as best sequences Borda score-wise in \cite{bouveret_voting_2017}. The results are presented in \autoref{tab:expPOA2}, which displays the mean value of $R_{ab}$ with Mallows for $\phi=1$ (Impartial Culture) and $\phi=0.6$. In both cases, for estimating the mean we sampled $10!$ (that is, around $10^6$) profiles. Our numerical experiment suggests that indeed increasing $n$ and $m$ leads to an increase of the mean and standard deviation. 

\begin{table}
\centering
\caption{Mean value of $R_{ab}(\pi,V)$ and standard deviation, for Impartial Culture and Mallows culture with $\phi=0.6$.}
\begin{tabular}{c|c|c|c}
$(n,m)$ & sequence & $\langle R_{ab} \rangle_{IC}$ & $\langle R_{ab} \rangle_{M_{0.6}}$  \\ 
\hline
(5,10) & 112321345  & $1.1 \pm 0.16$ 
& $1.03 \pm 0.07$  \\ 
(5,10) & 123114235  & $1.1 \pm 0.16$ 
& $1.03 \pm 0.07$ \\ 
(5,10) & 123451243  & $1.1 \pm 0.16$ 
& $1.03 \pm 0.12$\\ 
(5,10) & 111222345  & $1.05 \pm 0.16$ 
& $1.03 \pm 0.07$ \\ 
\end{tabular}
\label{tab:expPOA2}
\end{table}

To have a more in-depth understanding of how the culture influences $R_{ab}$ we exhibit in the Appendix \ref{appendix:subsec:poaVisualisation} in Figure \ref{appendix:fig:hist_mallows_poa} two overlapping histograms with Mallows culture and different dispersion parameters.

\subsection{Sincerity Ratio, in Practice}\label{subsec:sr_exp}

In this part we analyse the same sequences as in \autoref{subsec:poa_exp}. In \autoref{tab:expSR1} we present the mean ratio $R_{cb}$ with its standard-deviation, the maximum value reached, and the theoretical upper bound \eqref{sincerity-ratio}, denoted as U.B. We computed these values after enumerating all profiles. Clearly max $R_{cb}=SR(\pi)$.
It is interesting to note that in practice both sequences $\pi_1 = (1, 1, 1, 2, 2, 2, 1)$ and $\pi_2 = (1, 2, 2, 2, 1, 1, 1)$ reach different maxima for the ratio of sincerity (although, by \autoref{theorem:sincerity-ratio}, they have the same upper bound). \footnote{Also, since $\pi_1$ and $\pi_2$ are reversed one with respect to the other, the maximum reached value of $R_{cb}$ by one of them is the minimum of the other (see Remark 3).}

\begin{table}[t]
\caption{Mean value, standard-deviation, and maximum value of $R_{cb}(\pi,V)$ for some sequences.}
\begin{tabular}{c|c|c|cc}
$(n,m)$ & sequence & $\langle R_{cb} \rangle$ & $SR$ & U.B.\\ \hline
(2,8) & 1112221  & $1.02 \pm 0.07$ & $11/8$ & 11/8\\
(2,8) & 1222111  & $0.99 \pm 0.06$ & $10/8$ & 11/8 \\
(2,8) & 1122111  & $0.99 \pm 0.04$ & $9/8$ & 12/8\\
(3,7) & 123123  & $1.01 \pm 0.12$ & $14/7$ & 14/7\\  
(3,7) & 123321 & $1 \pm 0$ & $1$ & 14/7 \\  
(3,7) & 111223  & $1.05 \pm 0.18$ & $15/7$ & 15/7\\
(3,7) & 112233  & $1.01 \pm 0.15$ & $14/7$ & 14/7\\
\end{tabular}
\centering
\label{tab:expSR1}
\end{table}

The same shortcomings of the upper bound \eqref{sincerity-ratio} is also revealed by the fact that the sequence $\pi_3 = (1, 1, 2, 2, 1, 1, 1)$  has supposedly a higher upper bound than $\pi_2$ but that in fact the sincerity ratio of $\pi_3$ is inferior to the one of $\pi_2$.

As predicted by Corollary \ref{cor:palindromic}, the palindromic sequence gives a steady mean of 1 with zero standard-deviation.
A notable point is that for some of the sequences the mean ratio is less than 1. This means that, on average for these sequences, {\em strategic voting is more favorable to social welfare than sincere voting}. Of course, it is understood on the basis of Theorem \ref{firstTheorem} that if this impact is beneficial on average for a given sequence, it is on average detrimental for the reversed sequence, as illustrated in  \autoref{appendix:fig:hist_SR_pi1-pi2} of Appendix \ref{appendix:subsec:srVisualisation}.

The maximum value of $R_{cb}$ is often well above average, which means that (i) most of the time strategic voting is only mildly detrimental to social welfare, but (ii) for a few specific profiles it incurs a high loss. The histograms presented in \autoref{appendix:fig:hist_SR_pi1-pi2} in Appendix \ref{appendix:subsec:srVisualisation} support this claim.

Also it is worth mentioning that all sequences reaching the upper bound (U.B.) of $SR$ respect the structure detailed in Appendix \autoref{appendix:sec:struct}.

\autoref{tab:expSR2} displays the mean sincerity ratio with its standard deviation depending on the culture for a few sequences. Like previously, we computed these values after testing $10!$ profiles.

\begin{table}[h!]
\centering
\caption{Mean value of $R_{cb}(\pi,V)$ and its standard-deviation according to different cultures.}
\begin{tabular}{c|c|c|cc}
$(n,m)$ & sequence & $\langle R_{cb} \rangle_{IC}$ & $\langle R_{cb} \rangle_{M_{0.6}}$  &  U.B.\\ \hline
(5,10) & 112321345  & $1.04 \pm 0.19$ & $1 \pm 0.08$ & 39/10 \\  
(5,10) & 123114235  & $1.09 \pm 0.17$ & $1 \pm 0.06$ & 39/10 \\ 
(5,10) & 123451243  & $1.01 \pm 0.15$ & $0.98 \pm 0.1$ & 38/10\\
(5,10) & 111222345  & $1.05 \pm 0.2$ & $1 \pm 0.08$ & 39/10\\ 
\end{tabular}
\label{tab:expSR2}
\end{table}

As also seen in \autoref{tab:expSR1}, the average sincerity ratio is always close to 1, indicating that, except for a few profiles, the impact of strategic behavior on social welfare is low. And similarly as for the ratio $R_{ab}$, the ratio $R_{cb}$ gets closer to 1 with less standard deviation when $\phi$ decreases. We present in Appendix \ref{appendix:subsec:srVisualisation} \autoref{appendix:fig:hist_mallows_sr} a visualisation of this phenomena with an histogram of the first sequence of \autoref{tab:expSR2} for different value of $\phi$.

\section{Conclusion}\label{sec:conclu}

We have formalized voting by sequential elimination as a turn-taking game. Our take-home messages are:

\begin{enumerate}
\item The price of anarchy is in the order of the number of agents $n$ as is the upper bound \eqref{sincerity-ratio} of the  sincerity ratio. This is high, but we recall that such games are reasonable only for a small number of players \cite{bouveret_voting_2017}.
\item On average, strategic voting is much less detrimental to social welfare than the worst case scenario. 
This positive effect is amplified for voters with correlated preferences.
\item The outcome of all subgame-perfect Nash equilibria of a sequential elimination game is unique and has a simple characterization and is polynomial-time computable. 
\item Strategic behaviour sometimes increases social welfare (as compared to a sequential elimination game played sincerely).
\end{enumerate}

Another reason to temper message 1 is that, as we recall, we assumed complete knowledge and perfect rationality. Neither is realistic in practice. To know better what would happen in a more realistic context we would think of performing lab experiments.

\bibliographystyle{splncs04}
\bibliography{main}

\begin{thebibliography}{10}
\providecommand{\url}[1]{\texttt{#1}}
\providecommand{\urlprefix}{URL }
\providecommand{\doi}[1]{https://doi.org/#1}

\bibitem{Anbarci06}
Anbarci, N.: {Finite Alternating-Move Arbitration Schemes and the Equal Area
  Solution}. Theory and Decision  \textbf{61}(1),  21--50 (August 2006).
  \doi{10.1007/s11238-005-4748-9},
  \url{https://ideas.repec.org/a/kap/theord/v61y2006i1p21-50.html}

\bibitem{AnshelevichF0V21}
Anshelevich, E., Filos{-}Ratsikas, A., Shah, N., Voudouris, A.A.: Distortion in
  social choice problems: The first 15 years and beyond. In: Proceedings of the
  Thirtieth International Joint Conference on Artificial Intelligence, {IJCAI}
  2021. pp. 4294--4301 (2021)

\bibitem{AzizBLM17}
Aziz, H., Bouveret, S., Lang, J., Mackenzie, S.: Complexity of manipulating
  sequential allocation. In: Proceedings of the Thirty-First {AAAI} Conference
  on Artificial Intelligence. pp. 328--334 (2017)

\bibitem{AzizGT17}
Aziz, H., Goldberg, P., Walsh, T.: Equilibria in sequential allocation. In:
  Algorithmic Decision Theory - 5th International Conference, {ADT} 2017,
  Luxembourg, Luxembourg, October 25-27, 2017, Proceedings. pp. 270--283 (2017)

\bibitem{BoutilierR16}
Boutilier, C., Rosenschein, J.S.: Incomplete information and communication in
  voting. In: Handbook of Computational Social Choice, pp. 223--258 (2016)

\bibitem{bouveret_voting_2017}
Bouveret, S., Chevaleyre, Y., Durand, F., Lang, J.: Voting by sequential
  elimination with few voters. In: Proceedings of the 26th International Joint
  Conference on Artificial Intelligence. p. 128–134. IJCAI'17, AAAI Press
  (2017)

\bibitem{bouveret_general_2011}
Bouveret, S., Lang, J.: A general elicitation-free protocol for allocating
  indivisible goods. In: Proceedings of the Twenty-Second International Joint
  Conference on Artificial Intelligence. p. 73–78. IJCAI'11, AAAI Press
  (2011)

\bibitem{BouveretL14}
Bouveret, S., Lang, J.: Manipulating picking sequences. In: {ECAI} 2014 - 21st
  European Conference on Artificial Intelligence, 18-22 August 2014, Prague,
  Czech Republic - Including Prestigious Applications of Intelligent Systems
  {(PAIS} 2014). pp. 141--146 (2014)

\bibitem{BranzeiCMP13}
Br{\^{a}}nzei, S., Caragiannis, I., Morgenstern, J., Procaccia, A.D.: How bad
  is selfish voting? In: Proceedings of the Twenty-Seventh {AAAI} Conference on
  Artificial Intelligence (2013)

\bibitem{ChenL020}
Chen, X., Li, M., Wang, C.: Favorite-candidate voting for eliminating the least
  popular candidate in a metric space. In: The Thirty-Fourth {AAAI} Conference
  on Artificial Intelligence, {AAAI} 2020. pp. 1894--1901 (2020)

\bibitem{DesmedtE10}
Desmedt, Y., Elkind, E.: Equilibria of plurality voting with abstentions. In:
  Proceedings 11th {ACM} Conference on Electronic Commerce (EC-2010),
  Cambridge, Massachusetts, USA, June 7-11, 2010. pp. 347--356 (2010)

\bibitem{ElkindGRS15}
Elkind, E., Grandi, U., Rossi, F., Slinko, A.: Gibbard-satterthwaite games. In:
  Proceedings of the Twenty-Fourth International Joint Conference on Artificial
  Intelligence, {IJCAI} 2015. pp. 533--539 (2015)

\bibitem{farquharson_theory_1969}
Farquharson, R.: Theory of Voting. International standard, Yale University
  Press (1969)

\bibitem{GrossAX17}
Gross, S., Anshelevich, E., Xia, L.: Vote until two of you agree: Mechanisms
  with small distortion and sample complexity. In: Proceedings of the
  Thirty-First {AAAI} Conference on Artificial Intelligence. pp. 544--550
  (2017)

\bibitem{kalinowski_social_nodate}
Kalinowski, T., Nardoytska, N., Walsh, T.: A social welfare optimal sequential
  allocation procedure. In: IJCAI International Joint Conference on Artificial
  Intelligence (04 2013)

\bibitem{kalinowski_strategic_2013}
Kalinowski, T., Narodytska, N., Walsh, T., Xia, L.: Strategic behavior when
  allocating indivisible goods sequentially (2013)

\bibitem{KavnerXia21}
Kavner, J., Xia, L.: Strategic behavior is bliss: Iterative voting improves
  social welfare. CoRR  \textbf{abs/2106.08853} (2021)

\bibitem{kohler_class_1971}
Kohler, D.A., Chandrasekaran, R.: A class of sequential games. Operations
  Research  \textbf{19}(2),  270--277 (1971)

\bibitem{koutsoupias_worst-case_2009}
Koutsoupias, E., Papadimitriou, C.: Worst-case equilibria. In: Proceedings of
  the 16th Annual Conference on Theoretical Aspects of Computer Science. p.
  404–413. STACS'99, Springer-Verlag, Berlin, Heidelberg (1999)

\bibitem{mallows_ranking_1957}
Mallows, C.L.: Non-null ranking models. Biometrika  \textbf{44}(1-2),  114--130
  (1957)

\bibitem{Meir17}
Meir, R.: Iterative voting. In: Endriss, U. (ed.) Trends in Computational
  Social Choice, pp. 69--86 (2017)

\bibitem{TominagaTY16}
Tominaga, Y., Todo, T., Yokoo, M.: Manipulations in two-agent sequential
  allocation with random sequences. In: Proceedings of the 2016 International
  Conference on Autonomous Agents {\&} Multiagent Systems, Singapore, May 9-13,
  2016. pp. 141--149 (2016)

\bibitem{Walsh16}
Walsh, T.: Strategic behaviour when allocating indivisible goods. In:
  Proceedings of the Thirtieth {AAAI} Conference on Artificial Intelligence.
  pp. 4177--4183 (2016)

\bibitem{XiaC10a}
Xia, L., Conitzer, V.: Stackelberg voting games: Computational aspects and
  paradoxes. In: Proceedings of the Twenty-Fourth {AAAI} Conference on
  Artificial Intelligence (2010)

\bibitem{xia_strategic_2011}
Xia, L., Conitzer, V., Lang, J.: Strategic sequential voting in multi-issue
  domains and multiple-election paradoxes. In: Proceedings of the 12th {ACM}
  conference on {Electronic} commerce - {EC} '11. p.~179. ACM Press, San Jose,
  California, USA (2011)

\end{thebibliography}

\clearpage

\section*{Appendices}
\appendix

\section{Strategic \textit{and} Sincere voters} \label{appendix:sec:sincstratvoters}

We have seen a way to determine the outcome of a strategic vote with \autoref{firstTheorem} knowing the the preferences of each voters. We now examine the case of a vote with strategic and sincere voters.
Let us give an example of how this kind of vote unfolds itself.

\begin{example}
Let $n = 3$, $m = 4$, $V = (abcd, cbad, bcad)$ and $\pi=(1,2,3)$. Voters 1 and 3 are sincere and voter 2 is strategic. The game goes as follows: 1 eliminates candidate $d$, then 2 strategically eliminates candidate $b$ and finally voter 3 eliminates $a$. The winner of the election is $c$.
Note that if voter 2 had voted sincerely, candidate $b$ would have been the winner of the election.
\end{example}

We call $\pi_{sincere}$ the {\em subsequence} composed of all the sincere voters and $\pi_{strategic}$ the one composed of all strategic voters. For instance, if in a sequence $\pi=(1,2,3,4)$ the voters 2 and 3 are sincere then $\pi_{sincere}=(2,3)$ and $\pi_{strategic}=(1,4)$.

\begin{theorem} \label{appendix:secondTheorem}
The result of a vote by sequential elimination with both strategic and sincere voters can be found by extracting the subsequence with only sincere voters, executing it, then reversing the remaining subsequence and executing it sincerely.
\end{theorem}

\begin{proof}
To change the vote of a sincere voter, a strategic voter will have to eliminate a candidate that would have been otherwise eliminated by the sincere voter. Hence, the strategic voters have no incentive to change the sincere voters' vote. 

Let us extract the subsequence with only sincere voters and execute it, because the votes of the sincere voters will not be changed by those of the strategic voters (who have no interest in doing so). Now, only the strategic voters remain and we know that executing a subsequence with strategic voters is equivalent to reversing the sequence and executing it with sincere voting (by Theorem \ref{firstTheorem}).
\end{proof}

{\em Remark: } If two sequences have the same sincere subsequence and the same strategic subsequence, the same candidate is elected, irrespective of how the two subsequences are interleaved.

\begin{example}
Let $n = 4$, $m = 6$ and
$V = (abcdef, edcbaf,$ $fdebca, afecdb)$. Assume voters 1 and 3 are strategic, and voters 2 and 4 are sincere. 
If we wish to find the winner of the vote with the sequence $\pi_1=(1,2,3,4,4)$ we follow the \autoref{appendix:secondTheorem} procedure:
\begin{itemize}
    \item We sincerely execute the subsequence $\pi_{1-sincere}=(2,4,4)$ which eliminates $f,b$ and $d$.
    \item Then we take the remaining subsequence $\pi_{1-strategic}=(1,3)$, we reverse it which gives us $(3,1)$, which we execute sincerely, leading to the elimination of $a$ then $e$.
\end{itemize}
This gives us $c$ as the winner.

Now let us take the sequence $\pi_2=(2,4,4,1,3)$, we can see that the procedure will give the same result. Similarly for $\pi_3=(2,4,1,4,3)$, $\pi_4=(2,4,1,3,4)$, etc.
\end{example}

\section{Sincerity Ratio proof}\label{appendix:sec:SRproof}

For the sake of readability we repeat \autoref{theorem:sincerity-ratio}.

\setcounter{theorem}{2}
\begin{theorem}\label{appendix:theorem:sincerity-ratio}
With $n$ voters and $m$ candidates the sincerity ratio 
has the following upper bound:
\begin{equation}\label{appendix:sincerity-ratio}
SR(\pi) \leq \frac{O_{\rm{max}}+(n-1)(m-1)}{m}, 
\end{equation}
where $O_{\rm{max}}$ is the number of 
occurrences of the voter with the highest number of occurrences in the sequence.
\end{theorem}

\begin{proof} This proof being similar to the one of \autoref{theorem:price-anarchy}, we just highlight the main differences.
Let us consider $n$ voters with a profile preference $V$, $m$ candidates, and a sequence $\pi$ of
size $(m-1)$  for which the number of occurrences of the voter with the highest number of occurrences is $O_{\rm{max}}$. We want to determine
the maximum ratio $R_{cb}$ (Eq. \eqref{sincerityRatio})
over 
all sequences and preferences with fixed $n$, $m$ and $O_{\rm{max}}$.

For the same reason as in the previous proof, at least one of the voters ($x$) should prefer $b$ over $c$ ( $b \succ_{V_x} c$ ). 
Similarly as well, no matter which voter ($x$) prefers $b$ over $c$, to maximize the ratio we must increase 
the rank of $b$ and $c$ in her preferences.
However, the decrease should not prevent $c$ from being elected. This results in $r(c,V_x) = m-O^{(\pi)}_x$ (lower limit of the rank of an elected candidate according to corollary \ref{cor:nTimes}) 
and because $r(b,V_x) = r(c,V_x)-1$ we have $S_B(c,V_x) = S_B(b,V_x)-1 = O^{(\pi)}_x$.

Following the same steps as in the proof of \autoref{theorem:price-anarchy}, we now get an appropriate upper bound of $R_{cb}$: 
\begin{equation} \label{appendix:finalRatioSR}
\forall V, R_{cb}(\pi,V) \leq \frac{(n-1)(m-1)+O_{\rm{max}}}{m} \; .
\end{equation}
This yields the result \eqref{appendix:sincerity-ratio}.
\end{proof}
\

\section{General structure of an instance reaching the maximum sincerity ratio}\label{appendix:sec:struct}
After having exhibited a sequence reaching the sincerity ratio's upper bound \eqref{sincerity-ratio}, in the present subsection we specify the common features of all instances that reach the upper bound.

As previously seen, in the configurations that have the largest sincerity ratio, the candidate $c$ must be at rank $m-O_{\rm{max}}$, and the candidate $b$ at rank $m-O_{\rm{max}}-1$ for the voter with the highest number of occurrences in the sequence ($x$). Also, for every other voter ($i$), candidate $c$ must be at rank 1 and candidate $b$ at rank $m-O^{(\pi)}_i$.

To ensure that $c$ is elected in the sincere vote and $b$ in the strategic vote, the instance 
must meet certain conditions: 
\begin{itemize}
\item Sincere voting: since $r(c,V_x) = m-O_{\rm{max}}$, $c$ is preferred to $O_{\rm{max}}$ candidates in $V_x$. 
None of them should be eliminated by another voter, otherwise $x$ will eliminate $c$.\footnote{ $x$ voting $O_{\rm{max}}$ times, it is clear that in a sincere vote she will vote for her $O_{\rm{max}}$ least preferred candidates if they are not eliminated beforehand.}
\item Strategic voting: since  for all $i \neq x, r(b,V_i) = m-O^{(\pi)}_i$ there are $O^{(\pi)}_i$ less preferred candidates than $b$ in $V_i$.
None of them should be eliminated by another voter, otherwise $i$ will eliminate $b$.
\end{itemize}

Therefore, in the sincere vote, $x$ eliminates her $O_{\rm{max}}$ least preferred  candidates, to which $b$ does not belong [$r(b,V_x)=m-O_{\rm{max}}-1$]. This 
means that $b$ is eliminated by another voter. This voter must have a candidate least 
preferred than $b$ already eliminated to vote for $b$ [$\forall i\neq x$, $r(b,V_i)=m-O^{(\pi)}_i$, so $b$ never belongs to the $O^{(\pi)}_i$ least preferred candidates of voter $i$]. 

Moreover, in the strategic vote $x$ eliminates 
$c$ (every other voter has $c$ as their favorite candidate). And we know that strategic vote is 
equivalent to sincere vote with the sequence reversed (theorem \ref{firstTheorem}). Hence, with the 
sequence reversed, one voter different from $x$ must have eliminated one of $x$'s $O_{\rm{max}}$ least preferred candidates.

This means that $x$ and one other voter must have in common at least one candidate they both prefer less than $b$.
Because of the rank of $b$ in each voter's preferences, it is clear that the number of slots available for candidates least preferred than $b$ is $m$ (see Table \ref{appendix:table1}).
In the strategic vote $b$ is elected, which means that those $m$ slots contain $m-1$ different candidates. 
Thus, exactly one candidate is less preferred than $b$ for 2 voters.\footnote{For the record, one may remark that if every voter had different least preferred $O^{(\pi)}_i$ candidates,  then sincere and strategic vote would give the same outcome: the unique candidate which does not belong to one of the least preferred $O^{(\pi)}_i$ candidates.}
As shown above, one of these voters is $x$, let us denote the other one by $y$ and by $e$ the candidate they both prefer less than $b$.
This candidate $e$ should be eliminated by $x$ in the sincere vote and by $y$ in the strategic vote.

If these conditions are met, then the instance considered reaches the upper bound of the sincerity ratio  \eqref{sincerity-ratio}.

The order of the sequence of votes should be such that $x$ eliminates the candidate $e$ before $y$ does in the sincere vote but $y$ eliminates $e$ first in the 
strategic vote (which is a sincere vote with the sequence reversed). 

It is clear that the ranking of $e$ in the preferences of $x$ and $y$ has an impact on the orders of the votes of $x$ and $y$ in the sequence making it possible to reach the bound \eqref{sincerity-ratio}.

\begin{example}\label{appendix:example6}
Let us build a scenario to maximize the ratio 
\eqref{sincerityRatio} which reaches the upper bound \eqref{appendix:finalRatioSR}. Suppose we have 3 voters, 7 candidates and the voter with the most occurrences is voter 1, with 4 votes (voters 2 and 3 have 1 vote each).

Let us call $c_1$ the candidate winner of the sincere vote and $c_2$ the candidate winner of the strategic vote. To maximize the ratio \eqref{sincerityRatio}, the preferences should respect the configuration depicted in \autoref{appendix:table:example6a}.

\begin{table}
\centering
\caption{Preferences maximizing the ratio \eqref{sincerityRatio} for Example \ref{appendix:example6}.}
\begin{tabular}{cc}
Voter ($i$) & Preferences ($V_i$)\\ \hline
1 &  $\times \, c_2 \, c_1 \times \times \times \times $\\ 
2 & $c_1 \times \times \times \times \, c_2 \, \times$  \\
3 & $c_1 \times \times \times \times \, c_2 \, \times$  \\
\end{tabular}
\label{appendix:table:example6a}
\end{table}
Indeed, $c_2$ and $c_1$ have to be next to each other in the preference of the voter with the largest number of occurrences, and $c_2$ has to be at the lowest eligible rank for the two other voters.

Now to fill up the preferences in order for $c_2$ to be elected in the 
strategic vote and $c_1$ in the sincere vote we study each situation:\\
$\Rightarrow$ For the sincere vote, we want $c_1$ to be the winner.
To ensure that voter 1 does not eliminate $c_1$, voter 1 has to 
eliminate her least 4 preferred candidates herself. Afterwards, candidate $c_1$, 
being the preferred candidate of both voters 2 and 3, will be the winner of the sincere vote. \\
$\Rightarrow$ For the strategic vote we want $c_2$ to be the winner. We know thanks to the theorem \ref{firstTheorem} that the strategic vote is equivalent to a sincere vote with the sequence reversed. 
To ensure that voter 2 or 3 does not eliminate $c_2$, their least 
preferred candidate should not have been previously eliminated by someone else
(i.e. they must eliminate their least preferred candidate themselves in 
a sincere vote with the sequence reversed). This results in voters 2 and 
3 to have different least preferred candidates. The last condition is 
that voter 2 and 3 cannot both eliminate one of the least 4 
preferred candidates of voter 1 (otherwise voter 1 will eliminate $c_2$).
The corresponding preferences are represented in \autoref{appendix:table:example6b}. 

\begin{table}
\centering
\caption{preferences presented in \autoref{appendix:table:example6a}
 leading to elect $c_1$ in the sincere vote and $c_2$ in the strategic vote.}
\begin{tabular}{cl}
Voter ($i$) & Preferences ($V_i$)\\ \hline
1 & $c_3 \, c_2 \, c_1 \times \times \, c_4 \, \times $\\ 
2 & $c_1 \times \times \times \times \, c_2 \; c_3$\\
3 & $c_1 \times \times \times \times \, c_2 \; c_4$\\ 
\end{tabular}
\label{appendix:table:example6b}
\end{table}

The sequence $\pi$ has to respect two conditions: when reversed and executed sincerely, the voter 3 shall eliminate her worst candidate herself and in the sincere vote the voter 1 shall eliminate the worst candidate of voter 3.
The sequence $\pi = (1, 1, 2, 1, 3, 1)$ fulfills these conditions.

This instance will give as winner for the sincere vote the candidate $c_1$ with $S_B(c_1,V) = 16$. And the winner of the strategic vote will be the candidate $c_2$,  with $S_B(c_2,V) = 7$.

\end{example}

It is interesting to remark that the bound \eqref{appendix:finalRatioSR} is reached for a specific type of sequence described in details in \autoref{appendix:sec:struct}.
Let us give a concrete example with the following sequence :
\begin{equation}
\pi =  (\underbrace{x, \dots , x}_{O^{(\pi)}_x \text{ times}},\underbrace{i, \dots , i}_{O^{(\pi)}_i \text{ times}},  \dots , \underbrace{j, \dots , j}_{O^{(\pi)}_j \text{ times}} ) .
\end{equation}
where the voters are arranged from the voter with the highest number of occurrences ($x$) voting first to the voter with the least occurrences ($j$) voting last (hence here $O_{\rm{max}} = O^{(\pi)}_x \geq O^{(\pi)}_i \geq O^{(\pi)}_j$).

 We consider the preference profile depicted in \autoref{appendix:table1} (ordered from left to right from best to worst): $r(b,V_x)=m-O_{\rm{max}}-1$, and for all other voters (from $i$ to $j$), $r(b,V_i)=m-O^{(\pi)}_i ,\cdots,  r(b,V_j)=m-O^{(\pi)}_j$.

\begin{table}[h!]
    \centering
    \caption{Preferences in an instance reaching the sincerity ratio, $x$ is the voter with the largest number of occurrences ($O_{\rm{max}}$). We have $r(b,V_x)=m-O_{\rm{max}}-1$, and for all other voters (from $i$ to $j$) $r(b,V_i)=m-O^{(\pi)}_i , \cdots ,  \, r(b,V_j)=m-O^{(\pi)}_j$.}
    \begin{tabular}{cc}
        Voter & Preferences \\ \hline
        $x$ & $\dots$ $b\,c$ $\dots$ \\ 
        $i$ & $c$ $\dots$ $b$ $\dots$ \\ 
        $\vdots$ & $\vdots$ \\
        $j$ & $c$ $\dots$ $b$ $\dots$ \\  
        \end{tabular}
    \label{appendix:table1}
\end{table}

Also, a candidate should not appear with a ranking worse than $b$ in the preferences of more than one voter;
except for voter $j$ which has the same least preferred candidate as $x$.

With this preference profile and order of votes, 
$c$ is the winner of the sincere vote and $b$ is the winner of the strategic vote.  Besides, the corresponding Borda scores are 
$S_B(b,V) = \sum^n_{i=1} S_B(b,V_i) = \sum^n_{i=1} m - r(b,V_i) = 1 + \sum^n_{i=1} O^{(\pi)}_i = m $. And $S_B(c,V) = \sum^n_{i=1} 
S_B(c,V_i) = \sum^n_{i=1} m - r(c,V_i) = (n-1)(m-1) + m - ( m - O_{\rm{max}} )$.

This demonstrates that for this instance, the ratio 
\eqref{sincerityRatio} reaches the bound \eqref{appendix:finalRatioSR}.

Two remarks are in order here:

\textit{Remark 3:} If we reverse the sequence in an instance giving the maximum sincerity ratio it will give the minimum ratio of the reversed sequence. Indeed, if for a profile of preferences $V$ we have $S_B(c,V)/S_B(b,V)=\alpha$ as maximum ratio of sincerity for $\pi$, then for the reversed sequence: $S_B(b,V)/S_B(c,V)=\frac{1}{\alpha}$ will be the minimum ratio of sincerity reached.

\textit{Remark 4:}  At fixed number of candidates $m$, the largest value of the sincerity ratio \eqref{sincerity-ratio} is obtained  for $n=m-1$ and $O_{\rm{max}}=1$, i.e. when each voter votes once. The corresponding maximum value is $[1+(m-2)(m-1)]/m$.

\section{Experiment visualisation}\label{appendix:sec:expvisualisation}
\subsection{Price of Anarchy}\label{appendix:subsec:poaVisualisation}
To have a more in-depth understanding of how the culture influences $R_{ab}$ we exhibit in Figure \ref{appendix:fig:hist_mallows_poa} two overlapping histograms with Mallows cultures of different dispersion parameters $\phi$ for $(n,m)=(5,10)$, using the voting  sequence of the first row of \autoref{tab:expPOA2}. For each histogram we sampled $10!$ preference profiles.

The most significant feature of this histogram is the very large number of preference profiles for which $R_{ab}=1$, i.e., for which the strategic vote leads to elect the optimal candidate. Note in particular that the average $\langle R_{ab} \rangle$ is much lower than the price of anarchy, which in the present case is $PoA=38/9 \simeq 4.22$ (see also the results of Tables \ref{tab:expPOA1} and \ref{tab:expPOA2}).

In  \autoref{appendix:fig:hist_mallows_poa} we consider 
two different cultures with $\phi=0.6$ and 0.9.
The trend appears to be that, the smaller $\phi$ (i.e., the most correlated the voters' preferences), the smaller the standard-deviation (as expected).

\begin{figure}
  \centering
  \includegraphics[width=1\linewidth]{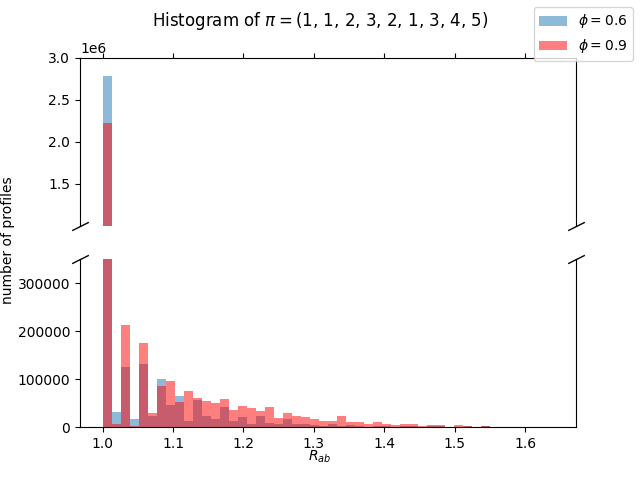}
  \caption{Histogram of the distribution of $R_{ab}$ with Mallows culture for different dispersion parameters.}
  \label{appendix:fig:hist_mallows_poa}
\end{figure}

\subsection{Sincerity Ratio}\label{appendix:subsec:srVisualisation}

In \autoref{appendix:fig:hist_SR_pi1-pi2} we display the histogram of $R_{cb}$ for the voting sequences of the two first rows of \autoref{tab:expSR1}. 
As noted in the main text, a sequence and its reverse give two $R_{cb}$ results which are the inverse of one another.
As a result, for one of the sequence the strategic vote being on average detrimental to social welfare, for the other one, it has the opposite effect and have instead a positive impact.

The histogram of \autoref{appendix:fig:hist_mallows_sr} presents the distribution of $R_{cb}$ for two Mallow cultures for the first sequence of \autoref{tab:expSR2}. We sampled $10!$ preference profiles. As already noticed from \autoref{appendix:fig:hist_mallows_poa} the standard variation decreases when the correlation between the voters' preferences increases.

\begin{figure}
  \centering
  \includegraphics[width=1\linewidth]{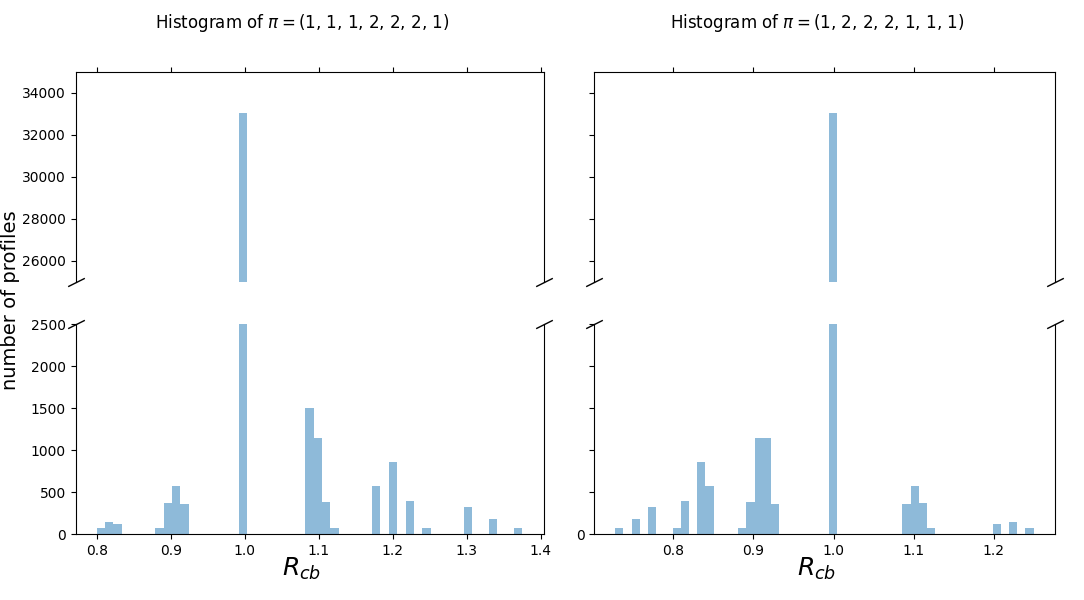}
  \caption{Distribution of the ratio $R_{cb}$ over all the 40320 profiles for the sequences $\pi_1$ and $\pi_2$.}
  \label{appendix:fig:hist_SR_pi1-pi2}
\end{figure}

\begin{figure}
  \centering
  \includegraphics[width=1\linewidth]{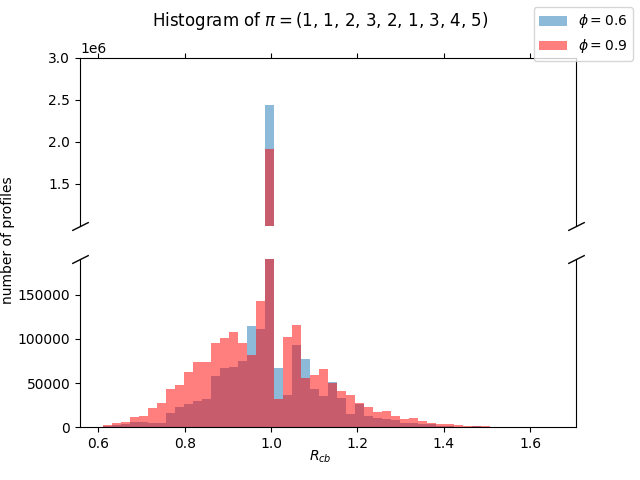}
  \caption{Distribution of the ratio $R_{cb}$ with Mallows culture for dispersion parameters $\phi=0.6$ and $0.9$.}
  \label{appendix:fig:hist_mallows_sr}
\end{figure}

\end{document}